\RequirePackage[T1]{fontenc}
\RequirePackage{fix-cm}
\RequirePackage{mlmodern}
\documentclass[12pt,pdftex]{article}

\usepackage{amsmath, amssymb}
\usepackage{amsthm}
\usepackage{bm}
\usepackage{url}
\usepackage{float}
\usepackage{fancyhdr}
\usepackage{geometry}
\usepackage{graphicx}
\usepackage{booktabs}
\usepackage{xcolor}
\usepackage{ulem}

\newtheorem{theorem}{Theorem}

\newtheorem{lemma}[theorem]{Lemma}

\newtheorem{remark}{Remark}

\newcommand{\ket}[1]{|#1\rangle}

\title{Construction of Quantum Rank-Metric Codes\\Using Hermitian Orthogonality}
\author{
  Ryota Nizuka \qquad Ryutaroh Matsumoto \\
  \small \textit{Department of Information and Communications Engineering} \\
  \small \textit{Institute of Science Tokyo}\\
  \small Tokyo, 152-8550 Japan.
}

\begin{document}
\maketitle
\begin{abstract}
Stacked quantum memory is an architecture in which multiple layers of qubits are stacked. 
Quantum rank-metric codes are effective for error correction in stacked quantum memories. 
However, the previously proposed quantum Gabidulin codes based on the CSS construction had a problem: due to algebraic constraints, the applicable memory layouts were strictly limited to square shapes of odd length.

In this paper, we first propose a framework for constructing quantum rank-metric codes from classical linear codes with symplectic self-orthogonality. 
Building upon this, we propose a new construction method for quantum Gabidulin codes by combining the Hermitian self-orthogonality of classical Gabidulin codes—utilizing the self-dual basis that exists when the extension degree of the finite field is even—with the quantum code construction method using Hermitian orthogonality by Matsumoto and Uyematsu.

The proposed method succeeds in approximately doubling the ratio of the minimum rank distance to the number of physical qubits while maintaining the code rate. 
Furthermore, it eliminates the restriction of the conventional method that requires the number of cells and layers of the stacked memory to be odd, realizing the construction of quantum rank-metric codes applicable to memories with an even number of cells and layers. 
This construction improves the relative error correction capability of the stacked quantum memory architecture and increases the degree of freedom in design while preserving the code rate.
\end{abstract}
\section{Introduction}\label{sec:introduction}
  Recent years have seen rapid progress in research towards the realization of quantum computers \cite{preskill2018, arute2019}. However, qubits are highly susceptible to noise, making quantum error correction essential for performing large-scale and reliable computations \cite{shor1996}.
  Among these, Delfosse and Z\'{e}mor proposed a model of stacked quantum memory, inspired by multi-level flash memory in classical SSDs, as an architecture to integrate physical qubits and build a space-saving system capable of large-scale computation \cite{qgab}.
  This stacked quantum memory model has a structure in which multiple layers of qubits are stacked, and a group of qubits penetrating vertically through each layer is treated as a ``cell''.
  Errors and noise occurring during the execution of Clifford circuits are modeled as stacked errors.
  To efficiently correct such errors, Delfosse and Z\'{e}mor proposed to extend classical rank-metric codes, particularly Gabidulin codes \cite{gabidulin1985, delsarte1978}, to quantum error correction.
  Delfosse and Z\'{e}mor constructed quantum Gabidulin codes using the CSS (Calderbank-Shor-Steane) construction \cite{calderbank1996, steane1996} and proposed error correction on stacked quantum memory.

  The quantum Gabidulin codes based on the CSS construction by Delfosse and Z\'{e}mor,\ require the existence of a self-dual normal basis of $\mathbb{F}_{2^n}$ over $\mathbb{F}_2$ \cite{macwilliams1977}, and are constructed when the number of physical qubits per layer, $n$, is odd.
  In addition, to satisfy the orthogonality required for the CSS construction, the applicable stacked quantum memories were strictly limited to $n \times n$ square shapes.
  
  On the other hand, research on the algebraic properties of classical Gabidulin codes has been advancing. Islam and Horlemann clarified the conditions under which Gabidulin codes become self-orthogonal with respect to the Hermitian inner product by using a self-dual basis of $\mathbb{F}_{2^n}$ over $\mathbb{F}_2$ when $n$ is even \cite{islam2023}.
  Matsumoto and Uyematsu proposed a method for constructing quantum codes for $p^m$-state systems from classical linear codes \cite{matsumoto2000}.

  In this paper, we propose a construction method for quantum Gabidulin codes using Hermitian orthogonality as an approach different from the CSS construction proposed by Delfosse and Z\'{e}mor.
  Specifically, by combining the Hermitian self-orthogonality of Gabidulin codes shown by Islam and Horlemann,\ and the quantum code construction method by Matsumoto and Uyematsu, we construct quantum Gabidulin codes with preferable characteristics.
  In this construction, while maintaining a code rate approximately equivalent to that of the conventional method, the ratio of the minimum rank distance to the number of physical qubits can be improved by a factor of about two.
  Also, the condition that the number of cells and layers of the applicable stacked memory must be odd, which was restricted in the conventional CSS construction, is removed.

  The structure of this paper is as follows.
  In Section 2, we review the fundamental concepts and theoretical frameworks underlying this research, including the definition of qubits, the model of stacked quantum memory, and the properties of Gabidulin codes.
  In Section 3, as one of the contributions of this paper, we describe a framework for constructing quantum rank-metric codes from classical codes with symplectic self-orthogonality.
  In Section 4, we review the quantum code construction method by Matsumoto and Uyematsu, which is the core of our proposed method.
  In Section 5, we review the conventional construction method by Delfosse and Z\'{e}mor, and examine the properties of the resulting quantum codes.
  In Section 6, the main result of this study, we propose a new construction method for quantum rank-metric codes using Hermitian orthogonality, and present its parameter evaluation and construction example.
  In Section 7, we compare the performance of the proposed method with the conventional method and discuss the superiority of the proposed method.
  Finally, Section 8 provides the conclusion of this paper.
\section{Preliminaries}\label{sec:preliminaries}
  In this section, we organize the concepts and theoretical frameworks that form the basis of the subsequent discussions.
  \subsection{Qubits}
    The minimum unit of information in quantum computation is called a qubit.
    While a classical bit takes a deterministic state of 0 or 1, a single qubit takes a pure state represented as a unit vector on a complex 2-dimensional Hilbert space $\mathcal{H}_2$.
    Assuming the orthonormal basis of $\mathcal{H}_2$ is $\{\ket{0}, \ket{1}\}$, the state $\ket{\psi}$ of a qubit is described as
    $\ket{\psi} = c_0\ket{0} + c_1\ket{1}$ 
    ($c_0, c_1 \in \mathbb{C}, |c_0|^2 + |c_1|^2 = 1$).
    Furthermore, the state of a composite system consisting of $n$ qubits is represented as a vector in the $2^n$-dimensional tensor product space $\mathcal{H}_2^{\otimes n}$.

    Errors and interactions with the environment that occur in quantum systems are formulated by unitary matrices on a Hilbert space, or trace-preserving quantum operations \cite{nielsen2010}.
    However, in quantum error correction theory, it is known that if a finite set of specific unitary matrices can be successfully corrected, the state can be sufficiently restored even against arbitrary continuous quantum errors occurring in the communication channel \cite{nielsen2010}.
    To represent this finite set of errors, we introduce the Pauli group.
    The Pauli matrices on a single qubit are defined by the following four matrices:
    \begin{equation}
      I := \begin{pmatrix} 1 & 0 \\ 0 & 1 \end{pmatrix},
      \quad X := \begin{pmatrix} 0 & 1 \\ 1 & 0 \end{pmatrix},
      \quad Z := \begin{pmatrix} 1 & 0 \\ 0 & -1 \end{pmatrix},
      \quad Y:=iXZ.
    \end{equation}
    Here, $X$ is the bit flip operator and $Z$ is the phase flip operator.
    The Pauli group $\mathcal{P}_n$ on an $n$-qubit system is defined as a finite group formed by multiplying the $n$-fold tensor products of these operators by phase factors $\{\pm 1, \pm i\}$.
    It is sufficient for a quantum error correction circuit to correct errors belonging to this Pauli group, and we will conduct our discussions focusing on these errors.
  \subsection{Stacked Quantum Memory}
    In this section, we formulate the mathematical model of stacked quantum memory and the errors that occur within it.
    \subsubsection{Stacked Quantum Memory and Circuit Application}
      Consider an $m \times n$ stacked quantum memory, which consists of $m$ physical qubits stacked vertically to form a ``cell'', with $n$ such cells arranged horizontally.
      As a result, the entire memory forms $m$ layers, and each layer contains $n$ physical qubits (Figure \ref{fig:stacked_memory}).
      The entire system consists of $mn$ physical qubits, and its state is described as a pure state $\ket{\psi}$ in the tensor product Hilbert space $\mathcal{H}_2^{\otimes mn}$.

      First, consider a Clifford circuit $C$ acting on a single layer ($n$ qubits).
      This circuit $C$ is defined as a composition of operations $U_1, U_2, \dots, U_s$.
      Here, each $U_i$ is originally a 1-qubit or 2-qubit Clifford gate, but by taking the tensor product with the identity operator for the unaffected qubits, it is treated as an operator applied to the entire $n$-qubit space.
      Next, consider a stacked implementation where this circuit $C$ is executed on all layers of the stacked quantum memory.
      This corresponds to applying $U_i$ simultaneously to all $m$ layers, and can be represented as applying the tensor product $U_i^{\otimes m}$ to the $mn$ qubits of the entire stacked quantum memory.
      Therefore, under ideal conditions without errors, the state after application becomes $\ket{\psi_{out}} = U_s^{\otimes m} \dots U_{2}^{\otimes m} U_1^{\otimes m} \ket{\psi}$.
      \begin{figure}[H]
        \centering
        \includegraphics[width=100mm]{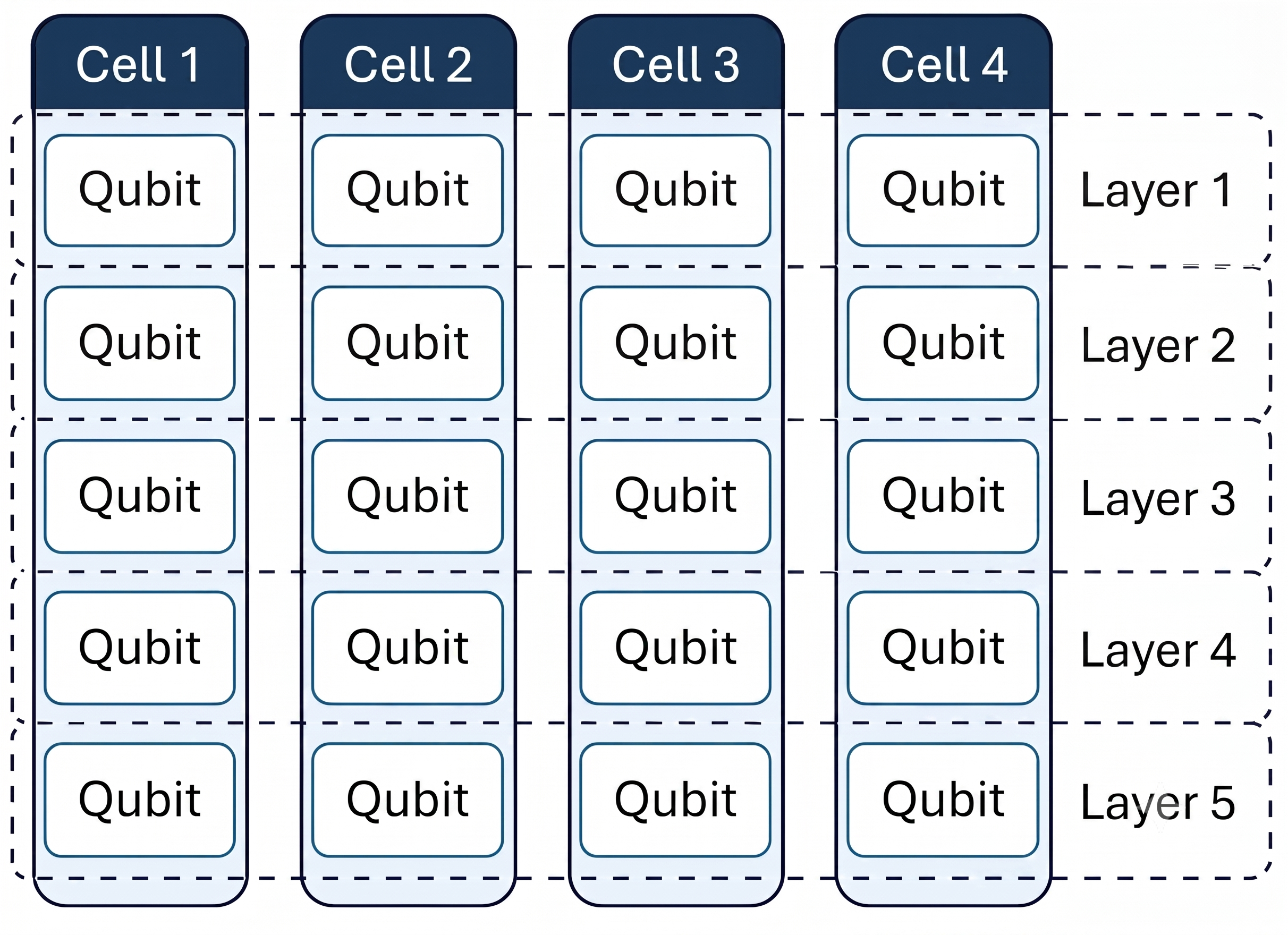}
        \caption{Conceptual diagram of a stacked quantum memory ($m=5, n=4$).}
        \label{fig:stacked_memory}
      \end{figure}
    \subsubsection{Stacked Errors and Definition of Rank}
      A Pauli error $P \in \mathcal{P}_{mn}$ occurring on the stacked quantum memory is viewed as an $m \times n$ matrix with components from $\{I, X, Z, Y\}$.
      Here, we first introduce a mapping $\mu_{n}: \mathcal{P}_n \to \mathbb{F}_2^{2n}$ that maps a Pauli operator $P_i \in \mathcal{P}_n$ in a single layer ($n$ qubits) to a row vector over the finite field $\mathbb{F}_2$, excluding the phase factor.
      When each component of the $i$-th layer is expressed as $P_{i,j} = X^{a_{i,j}}Z^{b_{i,j}} \;(a_{i,j}, b_{i,j} \in \mathbb{F}_2)$, we define $\mu_{n}(P_i) := (a_{i,1}, \cdots, a_{i,n} \mid b_{i,1}, \cdots, b_{i,n})$.
      Using this, we define a mapping $\mu_{m,n}: \mathcal{P}_{mn} \to \mathbb{F}_2^{m \times 2n}$ for the entire stacked error $P \in \mathcal{P}_{mn}$, which associates it with an $m \times 2n$ matrix formed by vertically arranging the row vectors $\mu_{n}(P_i)$ of each layer.
      \begin{equation}
        \mu_{m,n}(P):=
        \left(
            \begin{array}{ccc|ccc} 
              a_{1,1} & \cdots  & a_{1,n} & b_{1,1} & \cdots  & b_{1,n}\\ 
              \vdots  & \ddots & \vdots  & \vdots  & \ddots & \vdots\\
              a_{m,1} & \cdots  & a_{m,n} & b_{m,1} & \cdots  & b_{m,n}\\
            \end{array} 
            \right).
        \end{equation}

        We define the rank $\mathrm{rank}(P)$ of a stacked error $P \in \mathcal{P}_{mn}$ as the matrix rank of the matrix $\mu_{m,n}(P)$ over $\mathbb{F}_2$.
        \begin{equation}\label{eq:def_of_error_rank}
          \mathrm{rank}(P) := \mathrm{rank}_{\mathbb{F}_2}(\mu_{m,n}(P)).
        \end{equation}
        \begin{remark}
          In the literature \cite{qgab}, Delfosse and Z\'{e}mor define the rank of a stacked error $P$ as "the rank of the group generated by the error operators $P_{i,*}$ of each layer (the minimum number of generators required to generate the group)".
          However, since the Pauli group is a non-commutative group and includes phase factors $\{\pm 1, \pm i\}$, directly using a group-theoretic definition of rank creates a discrepancy with the rank of a matrix over $\mathbb{F}_2$ in classical rank-metric codes.
          Therefore, in this paper, to guarantee theoretical rigor in subsequent theorems, we redefine the rank of a quantum error as the rank over $\mathbb{F}_2$ of the $m \times 2n$ matrix $\mu_{m,n}(P)$ using the mapping $\mu_{m,n}$ which removes the phase components.
          This completely ensures the equivalence between the rank of a quantum error and the rank of a classical code.
        \end{remark}
    \subsubsection{Faulty Clifford Circuits}
      A ``fault'' in a gate refers to a situation where an unintended Pauli error is induced on the target qubits.
      Let the stacked errors that occur immediately after the application of each gate be $P^{(1)}, \dots, P^{(s)}$. Then, the state after applying the faulty circuit is as follows:
      \begin{equation}
        P^{(s)}U_s^{\otimes m} \dots P^{(2)}U_2^{\otimes m} P^{(1)}U_1^{\otimes m} \ket{\psi}.
      \end{equation}
      Since a single Clifford gate acts on at most one or two cells, the accompanying stacked error $P^{(t)}$ also acts on at most two cells.
      Therefore, the columns having non-zero components in the matrix $\mu_{m,n}(P^{(t)})$ are restricted to at most 4 columns out of the total $2n$ columns (the $X$ and $Z$ components of the affected cells).
      Thus, the rank of the error at each step is bounded by $\mathrm{rank}(P^{(t)}) \le 4$.
      Next, we present a lemma for the proof of the theorem.
      \begin{lemma}\label{lemma:rank}
        For a Clifford operator $U$ and a Pauli operator $P \in \mathcal{P}_{mn}$, the following holds:
        \begin{equation}
          \mathrm{rank}_{\mathbb{F}_2}(\mu_{m,n}(U^{\otimes m} P (U^{\otimes m})^\dagger))
          = \mathrm{rank}_{\mathbb{F}_2}(\mu_{m,n}(P)).
        \end{equation}
      \end{lemma}
      \begin{proof}
        We define a transformation $f: \mathbb{F}_2^{2n} \to \mathbb{F}_2^{2n}$ such that for any $P_s \in \mathcal{P}_n$,
        \begin{equation}f(\mu_{n}(P_s)) := \mu_{n}(U P_s U^\dagger).\end{equation}
        First, we show the well-definedness of the transformation $f$.
        From the definition of the mapping $\mu_{n}$, for any $P_1, P_2 \in \mathcal{P}_n$ satisfying $\mu_{n}(P_1) = \mu_{n}(P_2)$, it can be written as $P_2 = \lambda P_1 \ (\lambda \in \{\pm 1, \pm i\})$.
        Therefore,
        \begin{equation}
          U P_2 U^\dagger = U (\lambda P_1) U^\dagger = \lambda U P_1 U^\dagger.
        \end{equation}
        Applying the mapping $\mu_{n}$ to both sides yields $\mu_{n}(U P_2 U^\dagger) = \mu_{n}(U P_1 U^\dagger)$, which means $f(\mu_{n}(P_1)) = f(\mu_{n}(P_2))$, proving the well-definedness of $f$.

        Next, we show that the transformation $f$ has linearity over $\mathbb{F}_2$.
        For any $\vec{x}, \vec{y} \in \mathbb{F}_2^{2n}$, we take $P_x, P_y \in \mathcal{P}_n$ such that $\vec{x} = \mu_{n}(P_x), \vec{y} = \mu_{n}(P_y)$.
        From the property of the mapping $\mu_{n}$, $\mu_{n}(P_1 P_2) = \mu_{n}(P_1) + \mu_{n}(P_2)$, we have $\vec{x} + \vec{y} = \mu_{n}(P_x P_y)$, thus
        \begin{equation}
          \begin{aligned}
            f(\vec{x}+\vec{y}) 
            &= f(\mu_{n}(P_x P_y)) = \mu_{n}(U P_x P_y U^\dagger) = \mu_{n}(U P_x U^\dagger U P_y U^\dagger) \\
            &= \mu_{n}(U P_x U^\dagger) + \mu_{n}(U P_y U^\dagger) = f(\vec{x}) + f(\vec{y}),
          \end{aligned}
        \end{equation}
        which shows linearity.

        Next, we show that the transformation $f$ is invertible.
        For $f(\mu_{n}(P_s)) = \mu_{n}(U P_s U^\dagger)$, let $f'(\mu_{n}(P_s)) = \mu_{n}(U^\dagger P_s U)$.
        For any $\vec{x} \in \mathbb{F}_2^{2n}$ ($\vec{x} = \mu_{n}(P_s)$), evaluating the composite transformation yields:
        \begin{equation}
          \begin{aligned}
            f'(f(\vec{x})) &= f'(\mu_{n}(U P_s U^\dagger)) = \mu_{n}(U^\dagger (U P_s U^\dagger) U) \\
            &= \mu_{n}((U^\dagger U) P_s (U^\dagger U)) = \mu_{n}(P_s) = \vec{x}.
          \end{aligned}
        \end{equation}
        Therefore, an inverse mapping $f'$ exists for $f$, proving that $f$ is invertible.

        From the basic properties of linear algebra, an invertible linear transformation $f$ over the vector space $\mathbb{F}_2^{2n}$ can be uniquely described as a right multiplication using an invertible $2n \times 2n$ matrix $A$ over $\mathbb{F}_2$.
        That is, for any $P_s \in \mathcal{P}_n$, we can write $f(\mu_{n}(P_s)) = \mu_{n}(P_s) A$.

        Next, we evaluate the action on the entire stacked quantum memory.
        A Pauli operator $P \in \mathcal{P}_{mn}$ can be expressed using Pauli operators of each layer as $P = P_1 \otimes \dots \otimes P_m$ ($P_i \in \mathcal{P}_n$).
        The conjugate transformation by $U^{\otimes m}$ applied to this is:
        \begin{equation}
          U^{\otimes m} P (U^{\otimes m})^\dagger
          = (U P_1 U^\dagger) \otimes \dots \otimes (U P_m U^\dagger).
        \end{equation}
        From the definition of the mapping $\mu_{m,n}$, $\mu_{m,n}(P)$ is an $m \times 2n$ matrix where $\mu_{n}(P_i)$ is placed as the $i$-th row vector.
        Since the $i$-th row of the transformed matrix becomes $\mu_{n}(U P_i U^\dagger) = f(\mu_{n}(P_i)) = \mu_{n}(P_i)A$, the action on the entire matrix is expressed as
        \begin{equation}
          \mu_{m,n}(U^{\otimes m} P (U^{\otimes m})^\dagger) = \mu_{m,n}(P) A.
        \end{equation}
        Multiplying an invertible matrix to any matrix does not change its rank.
        Therefore, $\mathrm{rank}_{\mathbb{F}_2}(\mu_{m,n}(U^{\otimes m} P (U^{\otimes m})^\dagger)) = \mathrm{rank}_{\mathbb{F}_2}(\mu_{m,n}(P)A) = \mathrm{rank}_{\mathbb{F}_2}(\mu_{m,n}(P))$ holds.
      \end{proof}
      Using this lemma, a theorem similar to \cite[Lemma 2]{qgab} can be proven even under the definition of rank in this paper.

      \begin{theorem}\label{thm:faulty_circuit}
        Assume that $t$ gates out of $s$ Clifford gates are faulty.
        Then, all the stacked errors accumulated through the execution of the entire circuit can be aggregated into a single stacked error $Q$, and its rank satisfies $\mathrm{rank}(Q) \le 4t$.
      \end{theorem}
      \begin{proof}
        From its definition, a Clifford operation $U$ always has a Pauli operator $P'$ satisfying $U^{\otimes m} P (U^{\otimes m})^\dagger = P'$ for any stacked error $P$ \cite{gottesman1997}.
        By this property, we define $Q^{(t)}$ as the error obtained by applying conjugate transformations to the error $P^{(t)}$ at each step using the subsequent sequence of gate operations $U_s \dots U_{t+1}$ until the very end of the circuit.
        That is,
        \begin{equation}
          Q^{(t)} := (U_s^{\otimes m} \dots U_{t+1}^{\otimes m}) P^{(t)} (U_s^{\otimes m} \dots U_{t+1}^{\otimes m})^\dagger.
        \end{equation}
        Performing this for all $t$, the entire circuit operation can be transformed as follows, ignoring the phase factor $\lambda \in \{\pm 1, \pm i\}$:
        \begin{equation}
          \begin{aligned}
            P^{(s)}U_s^{\otimes m} \dots P^{(1)}U_1^{\otimes m} 
            = \lambda \left( Q^{(s)} \dots Q^{(1)} \right) U_s^{\otimes m} \dots U_1^{\otimes m}.
          \end{aligned}
        \end{equation}
        Here, we define the aggregated final stacked error as $Q := Q^{(s)} \dots Q^{(2)}Q^{(1)}$.
        By Lemma \ref{lemma:rank}, since $\mathrm{rank}(\mu_{m,n}(U^{\otimes m} P (U^{\otimes m})^\dagger))=\mathrm{rank}(\mu_{m,n}(P))$, we have
        \begin{equation}
            \mathrm{rank}_{\mathbb{F}_2}(\mu_{m,n}(Q^{(t)}))
            =\mathrm{rank}_{\mathbb{F}_2}(\mu_{m,n}(P^{(t)})).
        \end{equation}
        Furthermore, using the subadditivity of the matrix rank ($\mathrm{rank}(A+B) \le \mathrm{rank}(A) + \mathrm{rank}(B)$), we obtain the following inequality:
        \begin{equation}
          \begin{aligned}
            \mathrm{rank}(Q) &= \mathrm{rank}_{\mathbb{F}_2}(\mu_{m,n}(Q)) \\
            &= \mathrm{rank}_{\mathbb{F}_2}\left( \mu_{m,n}\left(\prod_{t=1}^s Q^{(s-t+1)}\right) \right) \\
            &= \mathrm{rank}_{\mathbb{F}_2}\left( \sum_{t=1}^s \mu_{m,n}(Q^{(t)}) \right) \\
            &\le \sum_{t=1}^s \mathrm{rank}_{\mathbb{F}_2}(\mu_{m,n}(Q^{(t)})) \\
            &= \sum_{t=1}^s \mathrm{rank}(P^{(t)}).
          \end{aligned}
        \end{equation}
        Considering that the number of faulty gates is $t$, and the rank induced by each gate fault is $\mathrm{rank}(P^{(t)}) \le 4$, it is immediately derived that $\mathrm{rank}(Q) \le 4t$.
      \end{proof}
  \subsection{Gabidulin Codes}
    In this section, we describe the definitions and basic performance of Gabidulin codes \cite{gabidulin1985, delsarte1978}, which are the classical rank-metric codes used in this study, as well as the Hermitian self-orthogonality crucial for constructing quantum codes.
    \subsubsection{Definition and Performance of Gabidulin Codes}
      Let $q$ be a prime power, $m$ be a positive integer, and $\mathbb{F}_{q^m}$ be the extension field of degree $m$ over $\mathbb{F}_q$.
      Also, take a vector $\vec{\alpha}=(\alpha_1, \dots, \alpha_m)\in \mathbb{F}_{q^m}^m$ consisting of $m$ linearly independent elements $\alpha_i\in \mathbb{F}_{q^m}$ over $\mathbb{F}_q$.
      Define a polynomial $f(X)$ of degree $q^{k-1}$ with coefficients $a_i \in \mathbb{F}_{q^m}$:
      \begin{equation}
        f(X) := b_0 X + b_1 X^q + b_2 X^{q^2} + \dots + b_{k-1} X^{q^{k-1}}.
      \end{equation}

      A Gabidulin code $\mathrm{Gab}(\vec{\alpha}, k)$ is defined as a linear space over $\mathbb{F}_{q^m}^m$ consisting of vectors obtained by evaluating all such polynomials $f$ at each element of $\vec{\alpha}$:
      \begin{equation}
        \mathrm{Gab}(\vec{\alpha}, k)
        := \{ (f(\alpha_1), f(\alpha_2), \dots, f(\alpha_m)) \mid b_0, \dots, b_{k-1} \in \mathbb{F}_{q^m} \}.
      \end{equation}

      Each codeword of this code can be regarded as an $m \times m$ matrix by expanding each component into a column vector over $\mathbb{F}_q$ using an arbitrary basis of $\mathbb{F}_{q^m}$ over $\mathbb{F}_q$, and the rank of this matrix is defined as the rank of the codeword.
      Using this, the rank distance $d_R(\vec{x}, \vec{y})$ between two vectors $\vec{x}, \vec{y}$ in $\mathbb{F}_{q^m}^m$ is defined as the rank of their difference vector $\vec{x}-\vec{y}$.
      Gabidulin codes have a minimum rank distance $d = m - k + 1$ and are called Maximum Rank Distance (MRD) codes \cite{gabidulin1985, delsarte1978}.
    \subsubsection{Gabidulin Codes with Hermitian Orthogonality}     
      We define the trace function from $\mathbb{F}_{q^m}$ to $\mathbb{F}_q$ as follows:
      \begin{equation}
        \mathrm{Tr} (x) := x + x^q + \dots + x^{q^{m-1}}.
      \end{equation}
      Then, using the Kronecker delta $\delta_{ij}$, a basis $\{\alpha_1, \dots, \alpha_m\}$ over $\mathbb{F}_{q^m}$ satisfying $\mathrm{Tr} (\alpha_i \alpha_j) = \delta_{ij}$ is called a self-dual basis.
      It is known that a necessary and sufficient condition for the existence of a self-dual basis over a finite field $\mathbb{F}_{q^m}$ is that either $q$ is even, or both $q$ and $m$ are odd \cite{wan2003}.
      When the extension degree $m$ is even, the Hermitian inner product $\langle \vec{x}, \vec{y} \rangle_H$ of two vectors $\vec{x} = (x_1, \dots, x_m)$ and $\vec{y} = (y_1, \dots, y_m)$ in $\mathbb{F}_{q^m}^m$ is defined as follows:
      \begin{equation}
        \langle \vec{x}, \vec{y} \rangle_H := \sum_{i=1}^m x_i y_i^{q^{m/2}}.
      \end{equation}
      \begin{lemma}[{Islam--Horlemann \cite[ Corollary 3.3]{islam2023}}]\label{thm:islam_hermitian}
        Let $m$ be even.
        Let $\vec{\alpha}$ be a vector consisting of elements that form a self-dual basis of $\mathbb{F}_{q^m}$ over $\mathbb{F}_q$. If the condition that the code dimension $k \le m/2$ is satisfied, then $C=\mathrm{Gab}(\vec{\alpha}, k)$ becomes self-orthogonal with respect to the Hermitian inner product.
        That is,
        $C \subseteq C^{\perp_H}$.
      \end{lemma} 
    
\section{Quantum Rank-Metric Codes}\label{sec:quantum_rm_code}
  In this section, we generalize and formulate the conditions and methods for constructing quantum rank-metric codes applicable to stacked quantum memories from classical codes, as proposed by Delfosse and Z\'{e}mor. \cite{qgab}.
  
  First, we define an isomorphic mapping $M:\mathbb{F}_2^{2mn}\to \mathbb{F}_2^{m\times 2n}$ that maps a vector of length $2mn$ to an $m \times 2n$ matrix.
  For $\vec{c}=(a_{1,1}, \dots, a_{1,n}, a_{2,1}, \dots,  a_{m,n},b_{1,1}, \dots,  b_{m,n})\in \mathbb{F}_2^{2mn}$, we define:
  \begin{equation}
    M(\vec{c}):=
        \left(
            \begin{array}{ccc|ccc} 
              a_{1,1} & \cdots  & a_{1,n} & b_{1,1} & \cdots  & b_{1,n}\\ 
              \vdots  & \ddots & \vdots  & \vdots  & \ddots & \vdots\\
              a_{m,1} & \cdots  & a_{m,n} & b_{m,1} & \cdots  & b_{m,n}\\
            \end{array} 
            \right).
  \end{equation}
  
  Next, using $\vec{a}=(a_1,\dots,a_n),\vec{b}=(b_1,\dots,b_n)\in\mathbb{F}_2^n$, we denote $(\vec{a}|\vec{b})=(a_1,\dots,a_n,b_1,\dots,b_n)\in\mathbb{F}_2^{2n}$.
  Using this notation, for binary vectors $(\vec{a}|\vec{b}), (\vec{a'}|\vec{b'})\in\mathbb{F}_2^{2n}$, we define the symplectic inner product $\langle (\vec{a}|\vec{b}), (\vec{a'}|\vec{b'})\rangle _{S}$ as follows:
  \begin{equation}
    \langle (\vec{a}|\vec{b}), (\vec{a'}|\vec{b'})\rangle _{S}
    := \langle\vec{a},\vec{b'}\rangle_{E} - \langle\vec{a'},\vec{b}\rangle_{E}.
  \end{equation}
  where $\langle\vec{a},\vec{b}\rangle_{E}$ represents the Euclidean inner product $\sum_{i=1}^n a_i b_i$.
  
  A commutative subgroup of the Pauli group is called a stabilizer group.
  We examine the procedure for constructing the stabilizer group of a quantum rank-metric code from a linear subspace $C \subset \mathbb{F}_2^{2mn}$ over a binary field.
  We define the set of Pauli matrices corresponding to the classical code $C$ as:
  \begin{equation}
    \mathcal{S} := \{ P \in \mathcal{P}_{mn} \mid M^{-1}(\mu_{m,n}(P)) \in C \}.
  \end{equation}
  A necessary and sufficient condition for two Pauli matrices to commute is that the symplectic inner product between their corresponding vectors is $0$ \cite{matsumoto2000,crss1997, matsumoto2021}.
  Therefore, if the classical code $C$ satisfies symplectic self-orthogonality $C \subset C^{\perp_S}$, any two elements of $\mathcal{S}$ commute with each other, and $\mathcal{S}$ fulfills the requirements as a stabilizer group.
  
  In this construction, the code space of the quantum code is defined as the simultaneous eigenspace corresponding to a specific set of eigenvalues for each element of the stabilizer group $\mathcal{S}$.
  The definition of codes using this simultaneous eigenspace follows the framework of previous studies \cite{crss1997, matsumoto2021}.
    
  When a quantum rank-metric code that encodes $K$ logical qubits using $N$ physical qubits and has a minimum rank distance of $D_R$ is denoted as $[[N, K, D_R]]$, the following theorem is obtained:

  \begin{theorem}\label{thm:general_qrm}
  If a linear subspace $C\subset \mathbb{F}_2^{2mn}$ has $|C|=2^{2nk}$, satisfies $C\subset C^{\perp_S}$, and 
  $d_R=\min_{\vec{c}\in C^{\perp_S}\setminus C} \mathrm{rank}_{\mathbb{F}_2}(M(\vec{c}))$ holds,
  then we can construct a $[[m\times n, n(m-2k), d_R]]$ quantum rank-metric code applicable to a stacked quantum memory of $m$ layers and $n$ cells.
\end{theorem}
\begin{proof}    
  The number of logical qubits $K$ is obtained from the dimension of the simultaneous eigenspace defined by the group $\mathcal{S}$.
  In a system with $mn$ physical qubits, the dimension of the simultaneous eigenspace corresponding to a symplectic self-orthogonal subspace of dimension $\dim_{\mathbb{F}_2} C = 2nk$ is $2^{mn - 2nk}$.
  Therefore,
  \begin{equation}
  K = \log_2(2^{mn - 2nk}) = mn - 2nk = n(m-2k)
  \end{equation}
  is obtained.
  The minimum distance of a stabilizer code is defined by the minimum weight of the Pauli operators that commute with the stabilizer group $\mathcal{S}$ but are not included in $\mathcal{S}$ itself.
  This set of operators corresponds to $C^{\perp_S} \setminus C$ in the vector space.
  Since the rank of a stacked error in this model is defined by definition (\ref{eq:def_of_error_rank}) and the rank of the matrix $M(\vec{c})$, the minimum rank distance is
  \begin{equation}
  d_R = \min_{\vec{c} \in C^{\perp_S} \setminus C} \mathrm{rank}_{\mathbb{F}_2}(M(\vec{c})).
  \end{equation}
\end{proof}
\section{Quantum Code Construction Method by Matsumoto and Uyematsu}\label{sec:mu_construction}
  In this section, we review the quantum code construction method using Hermitian self-orthogonality of classical codes by Matsumoto and Uyematsu, which is key to the construction method in this paper.
  
  Let $q$ be a prime or prime power, $n$ be a positive integer, and for an element $\alpha \in \mathbb{F}_{q^n}$, if $\{ \alpha, \alpha^q, \alpha^{q^2}, \dots, \alpha^{q^{n-1}} \}$ forms a basis of $\mathbb{F}_{q^n}$ over $\mathbb{F}_q$, this is called a normal basis.
  According to the Normal Basis Theorem \cite{lidl1997}, it is known that a normal basis exists for any finite extension field of a finite field.
  Therefore, hereinafter let $p$ be a prime, $n$ be a positive integer, and let $\{ \theta, \theta^p, \theta^{p^2}, \dots, \theta^{p^{2n-1}} \}$ be the normal basis of $\mathbb{F}_{p^{2n}}$ over $\mathbb{F}_p$.
  Next, for $\vec{v} = (a_1, \dots, a_n, b_1, \dots, b_n) \in \mathbb{F}_p^{2n}$, we define a mapping $\phi : \mathbb{F}_p^{2n} \to \mathbb{F}_{p^{2n}}$ as follows:
  \begin{equation}
    \phi(\vec{v}) := \sum_{j=1}^n a_j \theta^{p^{j-1}} + \sum_{j=1}^n b_j \theta^{p^{n+j-1}}.
  \end{equation}

  Furthermore, for $\vec{x}, \vec{y} \in \mathbb{F}_p^{2n}$, we define a linear mapping $T : \mathbb{F}_p^{2n} \times \mathbb{F}_p^{2n} \to \mathbb{F}_p$ as follows:
  \begin{equation}
    T(\vec{x},\vec{y}) := c_{n+1} - c_1.
  \end{equation}
  where the expansion of the product $\phi(\vec{x}) \phi(\vec{y})^{p^n}$ using the normal basis is given by $\phi(\vec{x}) \phi(\vec{y})^{p^n} = \sum_{i=1}^{2n} c_i \theta^{p^{i-1}}$ ($c_i \in \mathbb{F}_p$).
  Let the representation matrix of the linear mapping $T$ with respect to the standard basis of the vector space $\mathbb{F}_p^{2n}$ be denoted as $T$.
  That is, for any vectors $\vec{a}, \vec{b} \in \mathbb{F}_p^{2n}$, we let $T$ be the $2n \times 2n$ matrix satisfying $T(\vec{a}, \vec{b}) = \vec{a} T \vec{b}^T$.
  Here, $\vec{b}^T$ denotes the transpose vector of $\vec{b}$.

  Next, let $S = \begin{pmatrix} 0 & I_n \\ -I_n & 0 \end{pmatrix}$.
  Here, $I_n$ is the $n \times n$ identity matrix.
  Then, for the representation matrix of the linear mapping $T$, there always exists an invertible $2n \times 2n$ matrix $D$ over $\mathbb{F}_p$ that satisfies the following \cite[Lemma 10]{matsumoto2000}:
  \begin{equation}
    D T D^T = S.
  \end{equation}
    
  With the above preparations, we expand an extension field vector $\vec{c}=(c_1, \dots, c_m)\in \mathbb{F}_{p^{2n}}^m\quad(c_i \in \mathbb{F}_{p^{2n}})$ into a vector over the base field.
  For each component $c_i$, we compute $\vec{e}_i = \phi^{-1}(c_i)D^{-1}$ to obtain the row vector $\vec{e}_i \in \mathbb{F}_p^{2n}$.
  Let the components of $\vec{e}_i$ be $\vec{e}_i=(a_{i,1},\dots,a_{i,n},b_{i,1},\dots,b_{i,n})$.
  By doing this for each $i\in\{1,\dots,m\}$ and properly concatenating them, we define a base field vector $\Phi(\vec{c})\in\mathbb{F}_p^{2mn}$ as follows:
  $\Phi(\vec{c})=(a_{1,1}, \dots, a_{1,n}, a_{2,1}, \dots,  a_{m,n},b_{1,1}, \dots,  b_{m,n})$.

  Then, for $\vec{c},\vec{c'}\in\mathbb{F}_{p^{2n}}^m$, since $\langle \vec{c},\vec{c'}\rangle_H=0 \to \langle \Phi(\vec{c}),\Phi(\vec{c'})\rangle_S =0$ holds \cite[Proposition 11]{matsumoto2000}, we obtain the following theorem.
  \begin{theorem}[{Matsumoto-Uyematsu \cite[Theorem 12]{matsumoto2000}}]\label{thm:matsumoto_uyematsu}
    If a linear subspace $C\subset \mathbb{F}_{p^{2n}}^m$ is an $[m, (m-k)/2]$ code satisfying $C\subset C^{\perp_H}$, then $\Phi(C)\subset \Phi(C)^{\perp_S}$, and a $[[mn,kn]]_p$ quantum code can be constructed from $\Phi(C)$.
  \end{theorem}
\section{Delfosse-Z\'{e}mor Construction Method of Quantum Rank-Metric Codes}\label{sec:css_construction}
  In this section, we review the conventional Delfosse-Z\'{e}mor construction method for quantum rank-metric codes \cite{qgab}.
  
  First, let $n$ be an odd number, and consider a self-dual normal basis of $\mathbb{F}_{2^n}$ over $\mathbb{F}_{2}$, which possesses properties of both a self-dual basis and a normal basis:
  $\{ \alpha,\alpha^2, \alpha^{2^2}, \dots, \alpha^{2^{n-1}} \}
  \quad (\forall i,j,\quad \mathrm{Tr} (\alpha^{2^{i}} \alpha^{2^{j}}) =\delta_{i,j})$.
  It is known that such a self-dual normal basis always exists if the extension degree of the extension field is odd \cite{macwilliams1977}.
  
  Next, for this self-dual normal basis, let $\vec{\alpha}=(\alpha,\alpha^2, \alpha^{2^2}, \dots, \alpha^{2^{n-1}})$.
  We define the $2^r$-th power of $\vec{\alpha}$ as follows:
  \begin{equation}
    \vec{\alpha}^{2^r}=( \alpha^{2^r},\alpha^{2^{r+1}},  \dots,\alpha^{2^{n-1}},\alpha,\dots, \alpha^{2^{r-1}} ).
  \end{equation}

  Furthermore, we define the trace inner product for extension field vectors as follows.
  For $\vec{\beta}=(\beta_1,\dots,\beta_n), \vec{\gamma}=(\gamma_1,\dots,\gamma_n) \in \mathbb{F}_{q^m}^n$,
  \begin{equation}
    \langle\vec{\beta},\vec{\gamma} \rangle_{\mathrm{Tr}}
  :=\sum_{i=1}^n \mathrm{Tr}(\beta_i \gamma_i).
  \end{equation}

  Then, the following lemma holds.
  \begin{lemma}[{Delfosse-Z\'{e}mor \cite[Proposition 1]{qgab}}]\label{lemma-trace}
    The dual code of $\mathrm{Gab}(\vec{\alpha}, r)$ with respect to the trace inner product is $\mathrm{Gab}(\vec{\alpha}^{2^r}, n-r)$. That is, ${\mathrm{Gab}(\vec{\alpha}, r)}^{\perp_{\mathrm{Tr}}} =\mathrm{Gab}(\vec{\alpha}^{2^r}, n-r)$.
  \end{lemma}

  Next, using $r,s$ that satisfy $r+s<n$, we define $C_X=\mathrm{Gab}(\vec{\alpha}, r), C_Z=\mathrm{Gab}(\vec{\alpha}^{2^r}, s)$.
  Since $\mathrm{Gab}(\vec{\alpha}^{2^r}, s) \subset \mathrm{Gab}(\vec{\alpha}^{2^r}, n-r)$, $C_Z \subset C_X^{\perp_{\mathrm{Tr}}}$ holds from Lemma \ref{lemma-trace}.
  Using these codes $C_X, C_Z$ over the extension field, we show that a classical code $C$ over the binary field satisfying Theorem \ref{thm:general_qrm} in Chapter \ref{sec:quantum_rm_code} can be constructed.
  
  First, for an extension field vector $\vec{\beta} =(\beta_1,\dots,\beta_n)\in \mathbb{F}_{2^n}^n$, we define a mapping $\psi: \mathbb{F}_{2^n}^n \to \mathbb{F}_2^{n^2}$ that converts it into a binary field vector using the self-dual normal basis $\{ \alpha,\alpha^2, \dots, \alpha^{2^{n-1}} \}$ as follows:
  $\psi(\vec{\beta}):=(b_{1,1}, \dots, b_{1,n}, b_{2,1}, \dots,  b_{n,n})$
  ($\beta_i = \sum_{j=1}^n b_{i,j} \alpha^{2^{j-1}}$)
  
  Since $\{ \alpha,\alpha^2, \dots, \alpha^{2^{n-1}} \}$ is a self-dual basis, for any $\vec{\beta}, \vec{\gamma} \in \mathbb{F}_{2^n}^n$, the following relationship holds between the trace inner product over the extension field and the Euclidean inner product over the binary field after expansion:
  \begin{equation}\label{eq:trace_euclid_equivalence}
    \langle \vec{\beta},\vec{\gamma} \rangle_{\mathrm{Tr}} 
    = \langle \psi(\vec{\beta}), \psi(\vec{\gamma}) \rangle_{E}.
  \end{equation}
  
  Here, we construct a linear subspace $C \subset \mathbb{F}_2^{2n^2}$ over the binary field as follows:
  \begin{equation}
    C := \{ (\psi(\vec{\beta}) \mid \psi(\vec{\gamma}))\in \mathbb{F}_2^{2n^2} \,
    |\, \vec{\beta} \in C_X, \vec{\gamma} \in C_Z \}.
  \end{equation}
  
  Next, we verify that this code $C$ satisfies the condition $C \subset C^{\perp_S}$ of Theorem \ref{thm:general_qrm}.
  For any two codewords $\vec{c}_1 = (\psi(\vec{\beta}_1) \mid \psi(\vec{\gamma}_1)), \vec{c}_2 = (\psi(\vec{\beta}_2) \mid \psi(\vec{\gamma}_2)) \in C$, their symplectic inner product is calculated as follows:
  \begin{equation}
    \langle \vec{c}_1, \vec{c}_2 \rangle_S 
    = \langle \psi(\vec{\beta}_1), \psi(\vec{\gamma}_2) \rangle_E 
      - \langle \psi(\vec{\beta}_2), \psi(\vec{\gamma}_1) \rangle_E .
  \end{equation}
  
  By construction, $\vec{\beta}_1, \vec{\beta}_2 \in C_X$ and $\vec{\gamma}_1, \vec{\gamma}_2 \in C_Z$.
  Since $C_Z \subset C_X^{\perp_{\mathrm{Tr}}}$ holds as shown earlier, $\langle \vec{\beta}_1, \vec{\gamma}_2 \rangle_{\mathrm{Tr}} = 0$ and $\langle \vec{\beta}_2, \vec{\gamma}_1 \rangle_{\mathrm{Tr}} = 0$ are satisfied.
  Therefore, $\langle \vec{c}_1, \vec{c}_2 \rangle_S = 0$, which shows that $C \subset C^{\perp_S}$.
  Regarding the dimension of the code, since $|C_X| = 2^{nr}$ and $|C_Z| = 2^{ns}$, we have $|C| = 2^{n(r+s)}$.
  In particular, when setting $s = r$, we get $|C|= 2^{2nr}$.
  
  Furthermore, we evaluate the minimum rank distance $d_R$ of the code at this time.
  The symplectic dual space $C^{\perp_S}$ for the constructed code $C$ is expressed as follows from the properties of the mapping $\psi$ and Equation (\ref{eq:trace_euclid_equivalence}):
  \begin{equation}
    C^{\perp_S} 
    = \{ (\psi(\vec{x}) \mid \psi(\vec{z})) \in \mathbb{F}_2^{2n^2} 
    \,|\, \vec{x} \in C_Z^{\perp_{\mathrm{Tr}}}, \vec{z} \in C_X^{\perp_{\mathrm{Tr}}} \}.
  \end{equation}
  From Lemma \ref{lemma-trace}, $C_X^{\perp_{\mathrm{Tr}}} = \mathrm{Gab}(\vec{\alpha}^{2^r}, n-r)$.
  Due to the properties of classical Gabidulin codes, the minimum rank distance of non-zero codewords in this space is $n - (n-r) + 1 = r+1$.
  Similarly, applying Lemma \ref{lemma-trace} gives $C_Z^{\perp_{\mathrm{Tr}}} = \mathrm{Gab}(\vec{\alpha}^{2^{2r}}, n-r)$, and here too the minimum rank distance of non-zero codewords is $n - (n-r) + 1 = r+1$.
  Since the rank of a stacked error containing both $X$ and $Z$ is bounded from below by the rank of the respective $X$ error part and $Z$ error part, the minimum rank distance is $r+1$.
  From Theorem \ref{thm:general_qrm}, a quantum rank-metric code $[[n\times n, n(n-2r), r+1]]$ is constructed.
  This construction perfectly matches the quantum Gabidulin code $\mathrm{QGab}(\vec{\alpha}, r, r)$ proposed by Delfosse and Z\'{e}mor.
\section{Construction of Quantum Rank-Metric Codes Using the Matsumoto-Uyematsu Method}\label{sec:proposed_construction}
  In this section, using the theorems prepared in the previous sections, we construct a new quantum rank-metric code using Hermitian self-orthogonality and evaluate its performance.
  Specifically, by applying the Matsumoto-Uyematsu construction method to a classical Gabidulin code that is Hermitian self-orthogonal, we convert it into a binary vector space that satisfies the condition of Theorem \ref{thm:general_qrm} and is symplectic self-orthogonal.
  \subsection{Code Construction Procedure}
    Let the base field be $\mathbb{F}_2$, and consider an extension field $\mathbb{F}_{2^{2m}}$ with an extension degree of $2m$.
    It is known that a self-dual basis of the extension field always exists for the base field $\mathbb{F}_2$ \cite{wan2003}, and let the self-dual basis of $\mathbb{F}_{2^{2m}}$ over $\mathbb{F}_2$ be $\{\alpha_1, \dots, \alpha_{2m}\}$, and set $\vec{\alpha}=(\alpha_1, \dots, \alpha_{2m})\in\mathbb{F}_{2^{2m}}^{2m}$.
    For a positive integer $k$ satisfying $k < m$, we define the classical Gabidulin code $C = \mathrm{Gab}(\vec{\alpha}, k)\subset \mathbb{F}_{2^{2m}}^{2m}$.
    The code length of this code is $2m$, and the dimension is $k$.
    Since the extension degree $2m$ is even and satisfies $k < m = (2m)/2$, this code $C$ satisfies self-orthogonality with respect to the Hermitian inner product by Lemma \ref{thm:islam_hermitian}.
    That is, $C \subset C^{\perp_H}$ holds.
    Next, we apply the Matsumoto-Uyematsu mapping $\Phi: \mathbb{F}_{2^{2m}}^{2m} \to \mathbb{F}_2^{4m^2}$ to this classical code $C$ to obtain a linear subspace $\Phi(C)$ over the binary field (one should consider $m\to 2m, n\to m, p\to 2$ in the discussion of Section \ref{sec:mu_construction}).
    From the Matsumoto-Uyematsu theorem, if $C \subset C^{\perp_H}$ holds, the obtained binary vector space satisfies symplectic self-orthogonality.
    That is, $\Phi(C) \subset \Phi(C)^{\perp_S}$ holds.
    Since the dimension of the classical code $C$ over $\mathbb{F}_{2^{2m}}$ is $k$, the number of elements when considered as a vector space over $\mathbb{F}_2$ is $|C| = (2^{2m})^k = 2^{2mk}$.
    Because the mapping $\Phi$ is injective, $|\Phi(C)| = 2^{2mk}$.
    Therefore, the subspace $\Phi(C)$ satisfies all the conditions of Theorem \ref{thm:general_qrm} and constitutes the stabilizer group of a quantum rank-metric code applicable to a stacked quantum memory of $2m$ layers and $m$ cells.
  \subsection{Evaluation of Parameters and Derivation of Minimum Rank Distance}
In the proposed method, we use a stacked quantum memory with $2m$ layers and $m$ cells. 
  By substituting the parameters $m \to 2m$ and $n \to m$ into Theorem \ref{thm:general_qrm}, the number of elements is $|\Phi(C)| = 2^{2mk}$, which satisfies the condition of the theorem. 
  Therefore, the number of logical qubits $K$ of the resulting quantum code is
  \begin{equation}
    K = m(2m - 2k) = 2m(m - k).
  \end{equation}
  
  Next, we evaluate the minimum rank distance $d_R$ of this quantum code. For this evaluation, we first prepare the following lemma.
  \begin{lemma}\label{lemma:rank_equivalence}
  For any classical vector $\vec{c} \in \mathbb{F}_{2^{2m}}^{2m}$, its rank over $\mathbb{F}_2$ as a matrix obtained by expanding it into a binary vector by the Matsumoto-Uyematsu mapping $\Phi$, and then applying the matrix mapping $M$, matches the rank of the original classical vector. That is,
  \begin{equation}
    \mathrm{rank}_{\mathbb{F}_2}(M(\Phi(\vec{c}))) = \mathrm{rank}_{\mathbb{F}_2}(\vec{c}).
  \end{equation}
  holds.
  \end{lemma}
\begin{proof}
  The classical rank $\mathrm{rank}_{\mathbb{F}_2}(\vec{c})$ of a vector $\vec{c} = (c_1, \dots, c_{2m}) \in \mathbb{F}_{2^{2m}}^{2m}$ over the extension field is defined as the rank of the matrix obtained by expanding each component into a basis as described in Section 2.3.1.
  This is equivalent to the dimension of the subspace spanned by each component $c_i \in \mathbb{F}_{2^{2m}}$ over the base field $\mathbb{F}_2$, $\dim_{\mathbb{F}_2} \mathrm{span}_{\mathbb{F}_2}\{c_1, \dots, c_{2m}\}$.
  Define $f: \mathbb{F}_{2^{2m}} \to \mathbb{F}_2^{2m}$ as $f(x) = \phi^{-1}(x)D^{-1}$.
  Since the mapping $\phi$ is an isomorphism and the matrix $D$ is invertible, $f$ is an $\mathbb{F}_2$-linear isomorphism.
  Since a linear isomorphism preserves the dimension of a vector space, the dimension of the space spanned by the transformed row vectors $f(c_i)$ is equal to the dimension of the space spanned by the components $c_i$, i.e., the rank of $\vec{c}$.
  
  From the definition of the matrix mapping $M$ and the construction of $\Phi$, the $2m \times 2m$ matrix $M(\Phi(\vec{c}))$ representing the stacked quantum error is exactly the matrix formed by arranging these $f(c_i)$ vertically as the $i$-th row vectors.
  \begin{equation}
    M(\Phi(\vec{c})) = \begin{pmatrix} f(c_1) \\ \vdots \\ f(c_{2m}) \end{pmatrix}.
  \end{equation}
  Since the rank of a matrix equals the dimension of the space spanned by its row vectors,
  \begin{equation}
    \mathrm{rank}_{\mathbb{F}_2}(M(\Phi(\vec{c}))) 
    = \dim_{\mathbb{F}_2} \mathrm{span}_{\mathbb{F}_2}\{f(c_1), \dots, f(c_{2m})\} = \mathrm{rank}_{\mathbb{F}_2}(\vec{c}).
  \end{equation}
  holds.
\end{proof}

From the definition of the minimum rank distance $d_R$ in a quantum code, the minimum rank distance is given by the following equation:
\begin{equation}
  d_R = \min_{\vec{v} \in \Phi(C)^{\perp_S} \setminus \Phi(C)} \mathrm{rank}_{\mathbb{F}_2}(M(\vec{v})).
\end{equation}
The Matsumoto-Uyematsu mapping $\Phi$ has the property of preserving the duality of the inner product, such that $\Phi(C)^{\perp_S} = \Phi(C^{\perp_H})$ holds.
Therefore, any vector $\vec{v}$ belonging to the subspace $\Phi(C)^{\perp_S} \setminus \Phi(C)$ can be expressed as $\vec{v} = \Phi(\vec{c})$ using a classical vector $\vec{c} \in C^{\perp_H} \setminus C$.

Substituting this and applying Lemma \ref{lemma:rank_equivalence}, the equation for the minimum rank distance reduces to:
\begin{equation}
  d_R = \min_{\vec{c} \in C^{\perp_H} \setminus C} \mathrm{rank}_{\mathbb{F}_2}(M(\Phi(\vec{c}))) = \min_{\vec{c} \in C^{\perp_H} \setminus C} \mathrm{rank}_{\mathbb{F}_2}(\vec{c}).
\end{equation}

Since $C = \mathrm{Gab}(\vec{\alpha}, k)$ is a Maximum Rank Distance (MRD) code, its Hermitian dual code $C^{\perp_H}$ is also an MRD code, and its minimum rank distance is $2m - (2m - k) + 1 = k + 1$ \cite{islam2023}.
Thus, the minimum rank of $\vec{c} \in C^{\perp_H} \setminus C$ is $k + 1$.
From the above, $d_R = k + 1$ is derived, and it is proven that by applying the Matsumoto-Uyematsu construction method to an Hermitian self-orthogonal classical Gabidulin code, a new quantum rank-metric code $[[2m\times m, 2m(m - k), k+1]]$ can be constructed through an approach different from the CSS construction by Delfosse and Z\'{e}mor.
  \subsection{Construction Example for $m=2$, $k=1$}
  \begin{enumerate}
    \item Let the base field be $\mathbb{F}_2$, and construct the extension field $\mathbb{F}_{2^4}$ using the root $\omega$ of the irreducible polynomial $x^4 + x + 1 = 0$.
    \item Let the self-dual basis of $\mathbb{F}_{2^4}$ over $\mathbb{F}_2$ be $\{\omega^3, \omega^7, \omega^{12}, \omega^{13}\}$, and set $\vec{\alpha}=(\omega^3, \omega^7, \omega^{12}, \omega^{13})$.
    \item Construct the classical Gabidulin code $C = \mathrm{Gab}(\vec{\alpha}, 1) \subset \mathbb{F}_{2^4}^4$. From Lemma \ref{thm:islam_hermitian}, this code satisfies Hermitian self-orthogonality $C \subset C^{\perp_H}$.
    \item To apply the Matsumoto-Uyematsu mapping, set the normal basis of $\mathbb{F}_{2^4}$ to $\{\omega^3, \omega^6, \omega^{12}, \omega^9\}$ generated by $\theta = \omega^3$.
    \item Calculate the representation matrix $T$ of the mapping $T(x,y)$ based on this normal basis, to obtain
    $T=
          \begin{pmatrix}
          0 & 1 & 0 & 0 \\
          1 & 0 & 0 & 1 \\
          0 & 0 & 0 & 1 \\
          0 & 1 & 1 & 0
          \end{pmatrix}
    $.
    \item Calculate the transformation matrix $D$ satisfying $DTD^T = S$, obtaining
    $
      D = 
          \begin{pmatrix}
          1 & 0 & 0 & 0 \\ 
          0 & 0 & 1 & 0 \\ 
          0 & 1 & 0 & 0 \\ 
          1 & 0 & 0 & 1 
          \end{pmatrix}
    $.
    \item For each component of the code $C$, perform expansion using the normal basis and multiply by the matrix $D^{-1}$ to derive a linear subspace $\Phi(C)$ over the binary field.
    \item $\Phi(C)$ is a 4-dimensional subspace over $\mathbb{F}_2^{16}$, and the four independent generators of the corresponding stabilizer group $\mathcal{S} = \langle P_1, P_2, P_3, P_4 \rangle$ are obtained as follows:
    \begin{itemize}
      \item $P_1 = (X \otimes I) \otimes (Y \otimes X) \otimes (I \otimes X) \otimes (I \otimes Y)$
      \item $P_2 = (Z \otimes X) \otimes (X \otimes Y) \otimes (I \otimes Y) \otimes (Y \otimes Y)$
      \item $P_3 = (Y \otimes Z) \otimes (X \otimes Z) \otimes (Y \otimes Y) \otimes (Z \otimes Y)$
      \item $P_4 = (Z \otimes I) \otimes (X \otimes X) \otimes (Z \otimes Y) \otimes (I \otimes Z)$.
    \end{itemize}
  
  \item For the above stabilizer group $\mathcal{S}$, define the code space as the simultaneous eigenspace of $\mathcal{S}$.
  \end{enumerate}
      As a result, a $[[4 \times 2, 4, 2]]$ quantum rank-metric code encoding 4 logical qubits using 8 physical qubits (4 layers × 2 cells) with a minimum rank distance of 2 is constructed.
\section{Comparison with Conventional Methods}\label{sec:comparison}
    We compare the parameters and construction conditions of our method with those of the quantum Gabidulin codes based on the CSS construction proposed by Delfosse and Z\'{e}mor. \cite{qgab}.
  In the comparison, we define the code rate $R = K/N$ as the ratio of the total number of logical qubits $K$ to the total number of physical qubits $N$, and the relative rank distance $\delta= D/N$ as the ratio of the minimum rank distance to the total number of physical qubits $N$.
  Since the conventional method is limited to square-shaped memories of odd length, we compare the layout of $2n$ layers in our method with the conventional method's $2n-1$ and $2n+1$ layers, which are the preceding and succeeding layers respectively.
  For arbitrary positive integers $n, k$ ($n > k$), the parameters shown in Table \ref{tab:comparison_transposed} are obtained.
  \begin{table}[htbp]
    \centering
    \caption{Parameter Comparison Between the Conventional Method and Our Method}
    \label{tab:comparison_transposed}
    \resizebox{\textwidth}{!}{
    \begin{tabular}{llll}
        \toprule
        Parameter & Conventional Method \cite{qgab} & Conventional Method \cite{qgab} & Proposed Method \\ \midrule
        \begin{tabular}[c]{@{}l@{}}Physical Qubits $N$\\ ($=\text{layers}\times\text{cells}$)\end{tabular} & $(2n-1) \times (2n-1)$& $(2n+1) \times (2n+1)$ & $2n \times n$ \\ \addlinespace
        \begin{tabular}[c]{@{}l@{}}Logical Qubits $K$\\ ($=\text{layers}\times\text{cells}$)\end{tabular} & $(2n-1) \times (2n-2k-1)$& $(2n+1) \times (2n-2k+1)$ & $2n \times (n-k)$ \\ \addlinespace
        Minimum Rank Distance $D$ & $k+1$& $k+1$ & $k+1$ \\ \addlinespace
        Code Rate\\ $R=K/N$ & $1-\frac{2k}{2n-1}$& $1-\frac{2k}{2n+1}$ & $1-\frac{k}{n}$ \\ \addlinespace
        Relative Rank Distance\\ $\delta= D/N$ & $\frac{k+1}{(2n-1)^2}$& $\frac{k+1}{(2n+1)^2}$ & $\frac{k+1}{2n^2}$ \\ \bottomrule
    \end{tabular}
    }
\end{table}

  This construction demonstrates clear advantages over the conventional method in the following points.
  \begin{enumerate}
    \item Improvement in relative error correction capability\\
      From Table 1, we compare the proposed method and the conventional method by aligning their number of layers.
      Specifically, for the $2n$ layers of the proposed method, we consider cases of $2n-1$ and $2n+1$ layers in the conventional method.
      First, when the minimum rank distance $D = k+1$ is kept constant, it can be seen that the code rate $R = 1 - k/n$ of the proposed method maintains an almost equivalent value to $R = 1 - 2k/(2n \pm 1)$ of the conventional method.
      However, focusing on the relative rank distance $\delta = D/N$, $\delta$ of the proposed method achieves about twice the value of the conventional method.
      This means that when using the same number of physical qubits, the proposed method can achieve a higher capability to correct stacked errors.
      That is, based on Theorem \ref{thm:faulty_circuit}, it is possible to create a quantum code capable of handling a larger number of faulty Clifford gates.
    \item Applicability to even lengths\\
      The conventional method was restricted to square-shaped memories ($n \times n$) where the number of cells per layer $n$ is an odd number, due to the existence conditions of self-dual normal bases.
      In contrast, the proposed method eliminates this algebraic constraint by utilizing a self-dual basis.
      As a result, any positive number can be selected for the overall width $n$ of the stacked quantum memory.
  \end{enumerate}
  From the above results, the proposed method improves the error correction capability relative to the message while maintaining the code rate.
  At the same time, it expands the degree of freedom in designing stacked quantum memories, contributing to the construction of practical error correction.

\section{Conclusion}\label{sec:conclusion}
  In this paper, we proposed quantum rank-metric codes using Hermitian orthogonality as a new method to efficiently correct circuit faults in stacked quantum memories.

  First, as the foundation of this study, we formulated a general theoretical framework for constructing stabilizer groups with symplectic self-orthogonality from classical linear codes and extending them to quantum rank-metric codes.
  This enabled the construction of quantum rank-metric codes not restricted by conventional CSS constructions.
  In addition, to guarantee the effectiveness of quantum rank-metric codes, we reconsidered the definition of the rank of a stacked error, and strictly proved that the rank of the stacked error accumulated across the entire circuit is bounded from above by the number of faulty Clifford gates.
  
  Next, we examined the issue in the conventional quantum Gabidulin codes proposed by Delfosse and Z\'{e}mor, where the shape of applicable stacked quantum memories is restricted to square shapes of odd length due to the existence conditions of self-dual normal bases \cite{qgab}.
  To address this, in this study, we overcame this constraint by combining Hermitian self-orthogonal Gabidulin codes using self-dual bases that exist when the extension degree is even \cite{islam2023}, with the quantum code construction method utilizing the Hermitian inner product by Matsumoto and Uyematsu \cite{matsumoto2000}.
  
  With the proposed method, we succeeded in approximately doubling the minimum rank distance relative to the number of physical qubits while maintaining roughly the same code rate.
  As a result, for an equivalent number of physical qubits, it became possible to construct quantum codes capable of handling approximately twice as many faulty Clifford gates while preserving the code rate.
  Furthermore, it enabled stacked memory layouts where the number of layers and cells are even numbers.

\section*{Acknowledgment}
In accordance with IEEE guidelines on the use of AI-generated content, the authors disclose that an AI tool (Gemini) was utilized to assist with the English translation and content proofreading of this manuscript. Additionally, the conceptual diagram in Figure \ref{fig:stacked_memory} was generated using Gemini's image generation capabilities. The authors have carefully reviewed all AI-assisted content and take full responsibility for the final manuscript.

\end{document}